\documentclass{llncs}

\usepackage[dvipsnames]{xcolor}
\RequirePackage{amsmath}
\RequirePackage{physics}

\RequirePackage[noend]{algpseudocode}

\RequirePackage{amssymb}
\RequirePackage{braket}

\RequirePackage{amsthm}
\newtheorem{funct}{Functionality}

\RequirePackage{thmtools}
\RequirePackage{multirow}
\RequirePackage{stmaryrd}

\usepackage{hyperref}
\RequirePackage{cleveref}

\RequirePackage{calc}
\usepackage{dblfloatfix}

\RequirePackage{nicefrac}
\RequirePackage[textsize=tiny,textwidth=15mm,disable]{todonotes}
\RequirePackage{complexity}
\RequirePackage{tikz}
\usetikzlibrary{shapes,backgrounds,calc,arrows,automata,fit}
\usetikzlibrary{intersections,fit,shapes.misc, decorations.markings,patterns}
\usetikzlibrary{graphs,quotes}
\usetikzlibrary{arrows.meta}
\usetikzlibrary{decorations.pathreplacing}
\RequirePackage{soul}
\RequirePackage{wrapfig}

\usepackage{adjustbox}
\usepackage{array}
\usepackage{booktabs}
\usepackage{multirow}
\usepackage{empheq}
\usepackage{xspace}
\usepackage{setspace}
\usepackage{multirow}
\usepackage{makecell}

\usepackage{thmtools}
\usepackage{thm-restate}
\usepackage[shortlabels]{enumitem}
\usepackage{amssymb}
\newlist{todolist}{itemize}{2}
\setlist[todolist]{label=$\square$}
\usepackage{pifont}

\setlist[description]{font=\normalfont}

\usepackage{etoolbox}

\newcommand{\algstrut}[1][\algruledefaultfactor]{\vrule width 0pt
depth .25\baselineskip height #1\baselineskip\relax}
\newcommand*{\algrule}[1][\algorithmicindent]{\hspace*{.5em}\vrule\algstrut
\hspace*{\dimexpr#1-.5em}}

\makeatletter
\newcount\ALG@printindent@tempcnta
\def\ALG@printindent{\ifnum \theALG@nested>0\ifx\ALG@text\ALG@x@notext \else
    \unskip
    \ALG@printindent@tempcnta=1
    \loop
    \algrule[\csname ALG@ind@\the\ALG@printindent@tempcnta\endcsname]\advance \ALG@printindent@tempcnta 1
    \ifnum \ALG@printindent@tempcnta<\numexpr\theALG@nested+1\relax \repeat
    \fi
    \fi
}

\patchcmd{\ALG@doentity}{\noindent\hskip\ALG@tlm}{\ALG@printindent}{}{\errmessage{failed to patch}}

\AtBeginEnvironment{algorithmic}{\lineskip0pt}



\RequirePackage[bottom,symbol]{footmisc} 
\RequirePackage{footnote}
\makesavenoteenv{tabular}
\makesavenoteenv{table}

\newlength\mytemplen
\newsavebox\mytempbox

\makeatletter
\newcommand\mybluebox{\@ifnextchar[{\@mybluebox}{\@mybluebox[0pt]}}

\def\@mybluebox[#1]{\@ifnextchar[{\@@mybluebox[#1]}{\@@mybluebox[#1][0pt]}}

\def\@@mybluebox[#1][#2]#3{
    \sbox\mytempbox{#3}\mytemplen\ht\mytempbox
    \advance\mytemplen #1\relax
    \ht\mytempbox\mytemplen
    \mytemplen\dp\mytempbox
    \advance\mytemplen #2\relax
    \dp\mytempbox\mytemplen
    \colorbox{myblue}{\hspace{1em}\usebox{\mytempbox}\hspace{1em}}}
\makeatother

\definecolor{shadecolor}{cmyk}{.05,.05,0.05,0.05}
\definecolor{light-blue}{cmyk}{0.15,.15,.15,.15}
\newsavebox{\mysaveboxM} \newsavebox{\mysaveboxT}

\makeatletter
\def\blfootnote{\gdef\@thefnmark{}\@footnotetext}
\makeatother

\algnewcommand\algorithmicswitch{\textbf{switch}}
\algnewcommand\algorithmiccase{\textbf{case}}
\algnewcommand\algorithmicassert{\texttt{assert}}
\algnewcommand\Assert[1]{\State \algorithmicassert(#1)}\algdef{SE}[SWITCH]{Switch}{EndSwitch}[1]{\algorithmicswitch\ #1\ \algorithmicdo}{\algorithmicend\ \algorithmicswitch}\algdef{SE}[CASE]{Case}{EndCase}[1]{\algorithmiccase\ #1}{\algorithmicend\ \algorithmiccase}\algtext*{EndSwitch}\algtext*{EndCase}

\usepackage{enumerate}

\newcolumntype{R}[2]{>{\adjustbox{angle=#1,lap=\width-(#2)}\bgroup}l<{\egroup}}

\makeatletter
\patchcmd{\ALG@step}{\addtocounter{ALG@line}{1}}{\refstepcounter{ALG@line}}{}{}
\newcommand{\ALG@lineautorefname}{Line}
\makeatother

\usepackage{svg}
\usepackage{amsmath}
\definecolor{codegreen}{rgb}{0,0.6,0}
\definecolor{codegray}{rgb}{0.5,0.5,0.5}
\definecolor{codepurple}{rgb}{0.58,0,0.82}
\definecolor{backcolour}{rgb}{0.95,0.95,0.92}
\usepackage{pythonhighlight}
\makeatletter
\lstdefinestyle{python-style}{
    language=python,
    backgroundcolor=\color{backcolour},   
    commentstyle=\color{codegreen},
    keywordstyle=\color{magenta},
    numberstyle=\tiny\color{codegray},
    stringstyle=\color{codepurple},
    basicstyle=\ttfamily\footnotesize,
    breakatwhitespace=false,         
    breaklines=true,                 
    captionpos=b,                    
    keepspaces=true,                 
    numbers=left,                    
    numbersep=5pt,                  
    showspaces=false,                
    showstringspaces=false,
    showtabs=false,                  
    tabsize=2
}
\makeatother

\usepackage{anyfontsize}
\usepackage[bottom]{footmisc}
\usepackage{quantikz}
\usepackage[makeroom]{cancel}
\usepackage{forest}
\usepackage{wrapfig}
\usepackage{tabularx}

\usepackage[table]{xcolor}
\usepackage{lineno} 
\newcommand{\myparagraphMajor}[1]{{\par\vspace{0.4\baselineskip}\noindent\textbf{#1}~}}

\newcommand\vect[1]{\ensuremath{\left( #1 \right )}\xspace}

\definecolor{light-blue}{cmyk}{0.8,.15,.15,.15}

\newcommand\defaccr[2]{\newcommand#1{#2\xspace}}
\newcommand\defmath[2]{\newcommand#1{\ensuremath{#2}\xspace}}
\newcommand\concept[1]{\textit{#1}}

\defmath\PG{\mathcal{\hat P}_n}
\defmath\hP{{\hat P}}
\defmath\hQ{{\hat Q}}

\newcommand{\toolFormat}[1]{\texttt{#1}}
\defaccr{\quokka}{\toolFormat{Quokka\#}}
\defaccr{\quizx}{\toolFormat{QuiZX}}
\defaccr{\sliqsim}{\toolFormat{SliQSim}}
\defaccr{\sliqec}{\toolFormat{SliQEC}}
\defaccr{\gpmc}{\toolFormat{GPMC}}
\defaccr{\ganak}{\toolFormat{Ganak}}
\defaccr{\dfourmax}{\toolFormat{d4Max}}
\defaccr{\quasimodo}{\toolFormat{Quasimodo}}
\defaccr{\ddsim}{\toolFormat{DDSim}}
\defaccr{\qcec}{\toolFormat{QCEC}}
\defaccr{\autoQ}{\toolFormat{AutoQ}}

\renewcommand\dots{\makebox[.7em][c]{.\hfil.\hfil.}}

\renewcommand\phi{\varphi}

\defmath{\img}{\mathtt{image}}
\defmath{\apre}{\mathtt{\forall preimage}}

\defmath{\defn}{\,\triangleq\,}

\let\set\undefined

\providecommand{\set}[1]{\ensuremath{\left\lbrace #1 \right\rbrace}}

\providecommand{\sizeof}[1]{\ensuremath{\left\vert{#1}\right\vert}}

\defmath{\bool}{\ensuremath{\mathbb{B}}}
\defmath{\complex}{\ensuremath{\mathbb{C}}}
\defmath{\integers}{\ensuremath{\mathbb{Z}}}
\defmath{\conditionalind}{\mathrel{\text{\scalebox{1.07}{$\perp\mkern-10mu\perp$}}}}
\defmath{\dx}{\partial x}
\defmath{\ddx}{\sfrac{\partial}{\partial x}}
\defmath{\half}{\textstyle{\frac{1}{2}}}

\newenvironment{smallmat}{\left[\begin{smallmatrix}}{\end{smallmatrix}\right]}

\defmath\oh{\mathcal O}

\defmath\yy{\begin{smallmat}
    0 & y^*\\
    y & 0\\
\end{smallmat}}

\defmath\ww{\begin{smallmat}
      0 & y   \\
      y^* & 0  \\
  \end{smallmat}
}

\newcommand\no[1]{\overline{#1}}

\newcommand{\jfid}{\text{Jamio\l{}kowski fidelity}}

\usepackage[title]{appendix}

\newcommand{\ERROR}{\textcolor{Red}{\ding{53}}} 
\newcommand{\timeout}{TO}

\newcommand{\fast}[1]{\cellcolor{green!20}{\textbf{#1}}}
\newcommand{\CRASH}{\textcolor{orange}{\ding{115}}} 
\title{Quokka\#: Quantum Computing with \#SAT}

\author{
Jingyi Mei\inst{1}\orcidID{0000-0002-4665-9818}\thanks{Corresponding author.},
Dekel Zak\inst{1,2}\orcidID{0009-0001-0315-1539},
Muhammad Osama\inst{1}\orcidID{0000-0002-5023-5348},\\
Tim Coopmans\inst{1,2}\orcidID{0000-0002-9780-0949},
Alfons Laarman\inst{1}\orcidID{0000-0002-2433-4174}}
\authorrunning{Jingyi Mei, Dekel Zak, Muhammad Osama, Tim Coopmans, \& Alfons Laarman}
\institute{
Leiden University, 
Leiden, The Netherlands
\and Delft University of Technology, Delft, The Netherlands
\email{\{j.mei,m.o.mahmoud,a.w.laarman\}@liacs.leidenuniv.nl; \email{\{d.z.zak,t.j.coopmans\}@tudelft.nl}
}
}

\renewcommand{\orcidID}[1]{\smash{\href{http://orcid.org/#1}{\protect\raisebox{-1.25pt}{\protect\includegraphics{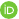}}}}}

\begin{document}

\maketitle

\begin{abstract}
We present \quokka, a versatile, open-source Python library for quantum circuit analysis.
\quokka reduces various simulation, verification, and synthesis tasks to weighted model counting (\#SAT). It supports universal quantum circuits and a wide variety of gates.
\quokka provides multiple encodings based on different algebraic bases and equi\-valence-checking methods, enabling key performance trade-offs.
Moreover, the new version of \quokka adds approximate equivalence checking,
which is crucial in its synthesis algorithms,
since it enables translation between arbitrary gate sets.
Its synthesis engine is depth-optimal, making it well-suited to real-world quantum computing.
This paper demonstrates the design, extensibility, and use of \quokka. 
\end{abstract}

\vspace{-1.5em}
\section{Introduction}
\label{sec:introduction}
\vspace{-.2em}

\begin{table}[b]
\renewcommand{\arraystretch}{1.5}
\setlength{\tabcolsep}{3pt}
\caption{
Functionalities (defined in \autoref{sec:func}) supported by \quokka and other state-of-the-art tools.
We denote non-universal and non-optimal circuit synthesis with
$\bigcirc$.
}
\vspace{-.2em}
\label{table:encoding-comparison}
\scalebox{0.58}{
\begin{tabularx}{1.7\textwidth}{l|c|c|c|c|c|c|c|c|c}  \hline\hline
\textbf{Functionality \textbackslash \ Tool}    &  {\quokka}    & \multicolumn{1}{|c|}{\toolFormat{MQT}}      & \sliqsim      & \sliqec       & \quasimodo    & \toolFormat{AutoQ} & \quizx & \toolFormat{Synthetiq} & \toolFormat{mtims} \\ 
\hline
\bf Weak / strong simulation 
    & $\cross$/$\checkmark$  
    & $\checkmark$/$\checkmark$ 
      \cite{zulehner2018advanced} 
    & $\checkmark$/$\checkmark$~\cite{tsai2021bit} 
    & $\cross$/$\cross$         
    & $\checkmark$/$\checkmark$~\cite{quasimodo} 
    & $\cross$/$\checkmark$~\cite{chen2023automata}   
    & $\cross$/$\checkmark$~\cite{kissinger2022simulating}    
    & $\cross$/$\cross$     
    & $\cross$/$\cross$     \\ \hline
\bf Eq. check. (exact / approx)
    & $\checkmark$ / $\checkmark$    
    & $\checkmark$/$\checkmark$
     \cite{advanced2021burgholzer,burgholzer2021random}    
    & $\cross$/$\cross$    
    & $\checkmark$/$\checkmark$\cite{wei2022accurate}   
    & $\checkmark$/$\cross$\cite{sistla2023weighted}
    & $\checkmark$/$\cross$\cite{chen2023automata}  
    & $\cross$/$\cross$     
    & $\cross$/$\cross$   
    & $\cross$/$\cross$  
    \\ \hline
\bf Partial eq. check. 
    & $\cross$    
    &  $\checkmark$
      \cite{advanced2021burgholzer}    
    & $\cross$     
    & $\checkmark$~\cite{chen2022partial}   
    & $\cross$     
    & $\cross$ 
    & $\cross$     
    & $\cross$ 
    & $\cross$   \\ \hline
\bf (Hoare logic) verification 
    & $\checkmark$       
    & $\cross$     
     
    & $\cross$     
    & $\cross$     
    & $\cross$ 
    & $\checkmark$~\cite{chen2023automata} 
    & $\cross$     
    & $\cross$ 
    & $\cross$  \\ \hline
\bf Synthesis (exact / approx.)     
    & $\checkmark$/$\checkmark$  
    & $\bigcirc$/$\cross$
     \cite{schneider2023sat}
    & $\cross$/$\cross$         
    & $\cross$/$\cross$             
    & $\cross$/$\cross$         
    & $\cross$/$\cross$     
    & $\cross$/$\cross$         
    & $\bigcirc$/$\bigcirc$~\cite{paradis2024synthetiq} 
    & $\checkmark$/$\cross$~\cite{amy2013mitm}
    \\ 
\hline\hline
\end{tabularx}}
\end{table}

Decades of research into satisfiability have yielded powerful solvers.
Initial attempts to harness these solvers for quantum circuit analysis achieved some success, but these results were also limited to non-universal circuits \cite{berent2022towards,schneider2023sat} or a few qubits~\cite{Bauer2023symQV,chen2023theory}.
By switching from SAT/SMT to weighted model counting (WMC or \#SAT), we recently enabled the use of state-of-the-art solvers for quantum circuit simulation \cite{mei2024simulating}, equivalence checking \cite{mei2024eq}, and synthesis~\cite{zak2025syntesis}.
Compared to previous work, the new \quokka library extends these advances by adding Hoare logic verification and approximate circuit equivalence (see \autoref{table:encoding-comparison}).
Moreover, it also adds a computational-basis encoding, which we empirically show to outperform the original Pauli-basis encoding~\cite{mei2024simulating,mei2024eq},
and introduces a new cyclic encoding enabling both exact and approximate equivalence checking.

\quokka currently supports the most commonly used (universal) quantum gate sets and allows users to add arbitrary gates. 
\quokka's performance is competitive with other state-of-the-art tools.
As \quokka lifts model counting from its traditional domain of classical probabilities~\cite{chavira2008probabilistic,shaw2024model,korhonen2023sharpsattd,chakraborty2015weighted} to the quantum domain, its impact has spilled back to the satisfiability community: the developers of d4Max~\cite{audemard2022maxsat} and Ganak~\cite{srsm19} have added support for complex numbers, and the \#SAT competition was extended with \quokka's benchmarks~\cite{Competition2021_23}.
The open-source implementation of \quokka has attracted interest from
international researchers and institutions as reflected by follow-up works~\cite{huang2026equivalence} and issue reports on its GitHub repository.
\todo[color=green!20]{Tim: needs proof. Point to (closed) github issues?}
These uses demonstrate the tool's future potential for verification researchers and symbolic reasoning experts to impact quantum computing research.

 \section{Background}
We provide the background on quantum computation and weighted model counting needed for this paper.
For an elaborate introduction, see \cite{quist2024advancing,nielsen2000quantum}.

\myparagraphMajor{Weighted model counting (WMC or \#SAT).}
The \emph{model count} of a  propositional formula $F(x_1, \dots, x_n)$ is the number of its satisfying assignments~\cite{valiant1979complexity,sang2004combiningcc}.
A weight function $W: L\rightarrow \mathbb{R}$ assigns weights to literals $L = \set{x_1, \dots, x_n, \overline{x_1}, \dots, \overline{x_n}}$.
An assignment's weight is the product of the weights of its literals, and the weighted model count
$\sizeof{F}_W$ is the sum of the weights of $F$'s satisfying assignments~\cite{sang2005performing,chavira2008probabilistic}.
As our computational-basis encoding requires complex numbers,
we extend the weight function to the complex domain, i.e.,\ $W:L\rightarrow\mathbb{C}$. 

\myparagraphMajor{Quantum Computing.}
An $n$-qubit \emph{quantum state} is a complex $2^n$-sized vector of unit norm, typically written in Dirac notation as $\ket{\dots}$.
We have $\ket0=\vect{1,0}^\top$, $\ket1=\vect{0,1}^\top$ and the \emph{computational-basis state} $\ket{b_1 b_2\dots b_n}$ represents the combined state of $n$ qubits, where qubit $j$ is in state $\ket{b_j}$ ($b_j \in \{0, 1\}$).
Any $n$-qubit quantum state $\ket{\psi}$ can be written as a linear combination of computational-basis states:  $\ket{\psi} = \sum_{b\in\{0,1\}^n}\alpha_b\ket{b}$.
A computational-basis measurement on this state returns the outcome $b \in \{0, 1\}^n$ with probability $|\alpha_b|^2$.
Another way to represent an $n$-qubit quantum state is by the $2^n \times 2^n$ \emph{density matrix} $\ket{\psi} \cdot \bra{\psi}$, where $\bra{\psi} = \left(\ket{\psi}\right)^{\dagger}$, 
which can be decomposed into \emph{Pauli basis}, a set of matrices that span the space of $2^n \times 2^n$ matrices.
A \emph{quantum gate} $U$ (a $2^n \times 2^n$ unitary matrix) updates a quantum state $\ket{\phi} \mapsto U\cdot \ket{\phi}$ where $\cdot$ denotes matrix-vector multiplication.
A \emph{quantum circuit} is a sequence of gates $C = (G_1,G_2,\cdots, G_m)$. \section{\quokka's Design}\label{sec:design}

\quokka is a Python library available on PyPI and at \cite{quokka_sharp_github}.
\autoref{fig:architecture} shows the design of the \quokka library, which comprises four engines for simulation, verification, equivalence checking, and synthesis, with their inputs and outputs.
\quokka understands quantum circuits in QASM format~\cite{cross2022openqasm}.
Its circuit encoder transforms the circuits to a Boolean formula with weights on literals as explained in \autoref{sec:encoding}.
The different engines instrument the encoding to implement the desired functionality (as detailed in \autoref{sec:func}).

\begin{figure}[t!]
\vspace{-3em}
\hspace{-.5cm}
\scalebox{.6}{

\begin{tikzpicture}[
  node distance=0.7cm and 1.8cm,
  every node/.style={font={\sffamily\large}},
  box/.style={rectangle, rounded corners=4pt, draw, thick, minimum height=0.95cm, align=center},
  encoder/.style={box, fill=pink!45, minimum width=2.5cm},
  post/.style={box, fill=cyan!10, minimum width=2.5cm},
  cnf/.style={box, fill=green!50!black!10, minimum width=2.5cm},
  solver/.style={box, fill=gray!8, minimum width=2.5cm},
  eqbox/.style={draw=orange!70!black, thick, dashed, rounded corners=4pt, inner sep=4pt},
  section/.style={font=\bfseries, align=center},
  inputlbl/.style={font=\small, color=blue!70!black},
  every chain/.style={start chain=going right},
encflow/.style={->, thick, draw=pink!70!black},
  postflow/.style={->, thick, draw=cyan!60!black},
  cnfflow/.style={->, thick, draw=green!50!black},
  solverflow/.style={->, thick, draw=gray!70!black},
  auxflow/.style={->, thick, dashed, draw=gray!70}
]

\node[encoder] (main_enc) at (0,0) {Circuit\\encoder};

\node[inputlbl, left=1cm of main_enc] (u_in) {\sffamily\large \textcolor{black}{Circuit $U$}};
\draw[encflow] (u_in) -- node[above]{\small QASM} (main_enc);

\node[post, right=1.5cm of main_enc, yshift=2.5cm] (post_sim) {Simulation\\engine};
\draw (main_enc.north east) edge[encflow, bend left=0, shorten <=-.1cm] (post_sim.south west);

\node[cnf, right =1cm of post_sim] (cnf_sim) {CNF\\generator};
\draw[postflow] (post_sim) -- (cnf_sim);

\node[solver, right=1cm of cnf_sim] (wmc_sim) {\#SAT\\solver};
\draw[cnfflow] (cnf_sim) -- node[above]{\small\sffamily wCNF} (wmc_sim);

\node[right=0.5cm of wmc_sim, font=\small\sffamily] (res_sim)
{probability $\bra{0^n}U^{\dagger}MU\ket{0^n}$};
\draw[solverflow] (wmc_sim) -- (res_sim);

\node[post, below=.8cm of post_sim, yshift=0cm] (post_ver) {Verification\\engine};
\draw (main_enc) edge[encflow, bend left=0] (post_ver);

\node[cnf, right=1cm of post_ver] (cnf_ver) {CNF\\generator};
\draw[postflow] (post_ver) -- (cnf_ver);

\node[solver, right=1cm of cnf_ver] (wmc_ver) {\#SAT\\solver};
\draw[cnfflow] (cnf_ver) -- node[above]{\small\sffamily wCNF} (wmc_ver);

\node[left=1.5cm of post_sim,font=\sffamily\large] (m) {Measurement specification $M$};
\draw[auxflow] (m.east) -- (post_sim);

\node[left=1.5cm of post_ver,font=\sffamily\large, yshift=.8cm] (c) {Pre- \& post-condition $P,Q$};
\draw[auxflow] (c.east) -- (post_ver);

\node[right=0.5cm of wmc_ver, font=\small\sffamily] (res_ver)
{true iff $\set{P}~U~\set{Q}$~\cite{li2019hoare}};
\draw[solverflow] (wmc_ver) -- (res_ver);

\node[post, below=.8cm of post_ver, yshift=0cm] (post_eq) {Equivalence\\engine};
\draw (main_enc) edge[encflow, bend left=0] (post_eq.north west);

\node[cnf, right=1cm of post_eq] (cnf_eq) {CNF\\generator};
\draw[postflow] (post_eq) -- (cnf_eq);

\node[solver, right=1cm of cnf_eq] (wmc_eq) {\#SAT\\solver};
\draw[cnfflow] (cnf_eq) -- node[above]{\small\sffamily wCNF} (wmc_eq);

\node[right=0.5cm of wmc_eq, font=\small\sffamily,text width=4cm] (res_eq)
{true iff $V$ is at least $1-\epsilon$ similar to $U$};
\draw[solverflow] (wmc_eq) -- (res_eq);

\node[encoder, left=1.5cm of post_eq] (v_enc) {Circuit\\encoder};
\draw[encflow] (v_enc) -- (post_eq);

\node[left=1cm of v_enc, blue] (v_in) {\sffamily\textcolor{black}{Circuit $V$}};
\draw[encflow] (v_in) -- node[above]{\small\sffamily QASM} (v_enc);

\node[post, below=.8cm of post_eq, yshift=0cm] (post_syn) {Synthesis\\engine};
\draw (main_enc.south east) edge[encflow, shorten <=-.1cm, shorten >=0cm, bend left=0] (post_syn.north west);

\node[cnf, right=1cm of post_syn] (cnf_syn) {CNF\\generator};
\draw[postflow] (post_syn) -- (cnf_syn);

\node[solver, right=1cm of cnf_syn] (wmax_syn) {Max\#SAT\\solver};
\draw[cnfflow] (cnf_syn) -- node[above]{\small\sffamily wCNF} (wmax_syn);

\node[right=0.5cm of wmax_syn, font=\small\sffamily,text width=4cm] (res_syn)
{depth-optimal circuit $V$\\in gate set $\mathcal G$ which is\\ at least $1-\epsilon$ similar to $U$};
\draw[solverflow] (wmax_syn) -- (res_syn);

\node[left=2.cm of post_syn, yshift=.8cm, font=\large\sffamily] (e)
{Precision parameter $\varepsilon \geq 0$};
\draw[auxflow] (e.east) -- (post_syn);
\draw[auxflow] (e.east) -- (post_eq);

\node[left=2.cm of post_syn, font=\large\sffamily] (g) {Gate set $\mathcal G$};
\draw[auxflow] (g.east) -- (post_syn);

\draw (wmax_syn.south west) edge[solverflow, bend left=10, font=\small\sffamily]
node[below] {\small iterate while $V \not \simeq_{1-\epsilon} U$} (post_syn.south east);

\draw[postflow, bend right=0]
  (post_syn) to node[midway, right, font=\small, text=cyan!60!black] {re-use} (post_eq);

\end{tikzpicture}

 }
\vspace{-1em}
    \caption{The architecture of \quokka.
}
    \label{fig:architecture}
\end{figure}
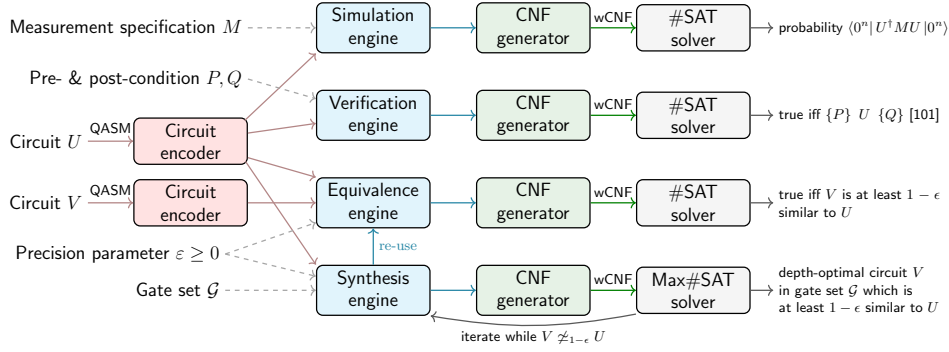

After instrumentation, the encoding is translated into weighted CNF (wCNF) in weighted DIMACS format~\cite{DCNF}.
The result is then passed to a (Maximum) \#SAT solver. 
We obtain the final result directly from the solver's output (for simulation and equivalence checking) or after a simple post-processing (for verification and synthesis).
For our Pauli-basis encoding, we initially extended the model counter GPMC~\cite{hashimoto2020gpmc} with negative weights~\cite{mei2024simulating}.
Since the computational-basis encoding needs complex numbers,
we extended GPMC with complex weights~\cite{GPMCc}.
We have actively collaborated with the model counting community to optimize model counters for quantum circuit analysis.
The Max\#SAT tool d4Max~\cite{lagniez2024d4} (for synthesis) and the state-of-the-art model counter Ganak~\cite{srsm19,SM2025} now both support complex weights.
In \autoref{sec:experiments}, we compare different counters.

\myparagraphMajor{Generality and Extensibility.}
\quokka supports the most common universal gate sets: Clifford+$T$~\cite{mei2024simulating}, Toffoli+$H$~\cite{mei2024eq}, and Clifford+$R$, where $R$ represents arbitrary rotation gates $R_X,R_Y$ and $R_Z$~\cite{mei2024simulating}.
Because we use \texttt{SymPy}~\cite{meuere2017sympy} to express the gate encoding in logic and then generate Python code with the corresponding CNF encoding, users can easily add arbitrary gates in either basis
by editing \texttt{comput2cnf\_py\_codegen.py} or \texttt{pauli2cnf\_py\_codegen.py}
(for encoding general unitaries, see~\cite{ende2025quantum}).
For example, as we shall see in \autoref{sec:encoding},
the encoding for the gate $X=\left[\begin{smallmatrix}
    0 & 1 \\
    1 & 0 
\end{smallmatrix}\right]$ is $F_X(q, q') = q' \Leftrightarrow \neg q$ in the computational basis. 
In \texttt{SymPy} syntax, this can be written as \texttt{Equivalent(q, \texttildelow q')}.

 \section{Encoding Quantum Computing in \#SAT}\label{sec:encoding}
Here, we exemplify how the \textsf{circuit encoder} (\autoref{fig:architecture})
represents circuits in the newly implemented computational-basis encoding.
See \cite{mei2024disentangling} for an in-depth theoretical trade-off analysis with Pauli-basis encoding~\cite{mei2024simulating,mei2024eq,quist2024advancing}.

\myparagraphMajor{Encoding Quantum States.\label{sec:encoding-states}}
In the computational-basis encoding, an $n$-qubit quantum state is represented
by a Boolean formula over $n$ variables $\vec q=\{q_1,\dots,q_n\}$, together
with a weight function $W$ and auxiliary variables when amplitudes differ from
$0$ or $1$.
For a state $\sum_{x\in\{0,1\}^n}\alpha_x\ket{x}$, each satisfying assignment
corresponds to a basis state $\ket{x}$, with its amplitude given by the
associated weight.
Basis states are encoded directly as conjunctions of literals.
For example, $\ket{01}$ is encoded as
$F_{\ket{01}}(q_1,q_2)=\neg q_1\wedge q_2$, whose unique satisfying assignment
has weight~$1$.
To encode superpositions, auxiliary variables are introduced to represent
non-unit amplitudes.
For instance, the state
$\ket{+}=\frac{1}{\sqrt{2}}(\ket{0}+\ket{1})$ is encoded by
$F_{\ket{+}}(q_1,h)=\no h$, together with a weight function $W_{\ket{+}}$ such
that $W_{\ket{+}}(\no h)=\frac{1}{\sqrt{2}}$.
Hence, a quantum state is represented by a pair $(F,W)$.
Based on the way our circuit encoding works, we can let $F_{\bra{\phi}} = F_{\ket{\phi}}$ by only conjugating the weight function of $F_{\ket{\phi}}$ for the involved literals.

\def\vbefore{\vec{v}_{\textnormal{before}}}
\def\vafter{\vec{v}_{\textnormal{after}}}
\def\qbefore{{q}_{\textnormal{before}}}
\def\qafter{{q}_{\textnormal{after}}}

\myparagraphMajor{Encoding Quantum Gates and Circuits. \label{sec:encoding-gates}}
The encoding $F_U$ of a gate $U$ satisfies $F_{U\ket{\phi}} (\vec{q}') = F_{\ket{\phi}}(\vec{q}) \wedge F_U (\vec{q}, \vec{q}')$ for all states $\ket{\phi}$, e.g. the Hadamard gate
$H = \frac{1}{\sqrt{2}}\begin{smallmat}
    1 & 1 \\
    1 & -1
\end{smallmat}$ 
is encoded as $F_H(q, q', h) = h \leftrightarrow (q \wedge q')$, with $W(\no h) = \frac{1}{\sqrt{2}}$ and $W(h) = -\frac{1}{\sqrt{2}}$.
A circuit is encoded by conjoining the encodings of its gates, much like in bounded model checking~\cite{biere2009bounded}, i.e.: $F_{G_1}(\vec q_0, \vec q_1) \land F_{G_2}(\vec q_1, \vec q_2) \land \dots \land F_{G_k}(\vec q_{m-1}, \vec q_m)$, where variables $\vec q_i$ represent the state after gate $G_i$.

 \section{Functionalities of \quokka}
\label{sec:func}

This section formalizes the functionality of the four \quokka modules shown in \autoref{fig:architecture} and describes their implementation.

\subsection{Simulation}
Classical simulation of quantum computing~\cite{viamontes2003improving,Anders2006fast,zulehner2018one,burgholzer2020improved} is a basic task in quantum compilation.
Due to the task's inherent complexity, good solutions are highly non-trivial
and have fundamental implications~\cite{bravyi2016trading,bravyi2019simulation}.
\quokka supports strong simulation, but not weak simulation (sampling a measurement's outcome distribution).
Here, we define a minimal notion of strong simulation, which suffices for supporting measurement specifications in the computational and Pauli bases.
\begin{funct}[Strong simulation]
    Given an $n$-qubit quantum circuit $C$
    and a measurement specification $P \in \{I, \ketbra{0}{0}, \ketbra{1}{1}\}^{\otimes n}$ (in the computational basis) or $P \in \{I, X, Y, Z\}^{\otimes n}$ (in the Pauli basis), compute the probability $\bra{\vec 0}C^\dagger P C\ket{\vec 0}$.
\end{funct}

Consider, for example, a quantum circuit consisting of a single Hadamard gate and a projector $P=\ketbra{0}$. \quokka encodes this query as:  
\[
F_{\ket{0}}(\vec q) \wedge F_H(\vec q, \vec q') \wedge F_{P}(\vec q, \vec q') \wedge  F_H(\vec q, \vec q') \wedge F_{\bra{0}}(\vec q'),
\]
where $F_{\ket{0}}$, $F_{P}(\vec q, \vec q')$, and $F_H$
are defined in \autoref{sec:encoding}.
The weighted model count of this encoding is $\frac{1}{{2}}$ because $|\langle 0 | H | 0\rangle|^2 = \frac{1}{{2}}$

To compute the probability of outcome 11 when measuring the first two qubits in the computational basis, we can use a projector $P =\ketbra{11}\otimes I^{\otimes n-2}$. \quokka then encodes this query as:
\begin{equation}\label{eq:sandwich}
 F_{\ket{0}}(\vec q_0) \wedge 
F_{C^\dagger}(\vec q_0, \vec q_1) \wedge 
F_{P}(\vec q_1, \vec q_2) \land
F_{C}(\vec q_2, \vec q_3) \wedge 
F_{\bra{0}}(\vec q_3).  
\end{equation}
\todo[color=green!20]{Tim: MAJOR: which projectors are allowed? How does the user construct them? Can it be any projector? (Note that the definition of a project is very general: hermitian matrix $O$ that squares to itself.}

 \subsection{Verification \label{sec:verification}}

We consider quantum Hoare logic as introduced in~\cite{li2019hoare},
restricted to unitary quantum circuits without measurements or control-flow constructs.
In this setting, quantum Hoare triples are of the form $\{P\}~C~\{Q\}$. The goal is to check whether executing circuit $C$ from a state in the pre-condition subspace, defined by projector $P$, yields a state in    the post-condition subspace, defined by $Q$.
\begin{funct}[Verification~\cite{li2019hoare}]\label{func:verify}
  Given a circuit $C$, projectors $P$ and $Q$, 
  the Hoare triple $\{P\}~C~\{Q\}$ is true if and only if 
  for any input state $\ket{\psi}$, we have
$
\ket{\psi} \models P \Longrightarrow C\ket{\psi} \models Q,
\text{ where: }
\ket{\psi} \models P \text{ iff } P\ket{\psi} = \ket{\psi}.
$
\end{funct}
\newcommand{\im}[1]{\operatorname{im}(#1)}

\begin{proposition}
Let $C$ be a unitary circuit and let $P,Q$ be projectors. Then the Hoare triple
$\{P\}~C~\{Q\}$ holds if and only if $\Tr(CPC^\dagger Q)=\Tr(P).$
\end{proposition}
\begin{proof}
By definition, $  \{P\}~C~\{Q\}  $ means $  QCP=CP  $.

$  (\Rightarrow)  $: $  QCP=CP  $ $  \implies  $ $  C^\dagger QCP=P  $ $  \implies  $ $  \Tr(C^\dagger QCP)=\Tr(P)  $. By cyclicity,
$\Tr(CPC^\dagger Q)=\Tr(P).$

$  (\Leftarrow)  $: Let $  A=CPC^\dagger  $. Because $C$ is a unitary and $P$ a projector, $A$ is also a projector.
Assume that $  \Tr(A Q)=\Tr(A)$.
Then $\Tr(A(I-Q))=0.$
By cyclicity and idempotence of projectors ($  A^2=A  $), we have that 
$\Tr(A(I-Q))=\Tr(A(I-Q)A).$
The fact that $  A(I-Q)A = ((I-Q)A)^\dagger ((I-Q)A) \succeq 0  $ and has trace zero, implies that $  (I-Q)A=0  $ (a property of positive-semidefinite operators). Thus $  A=QA  $, i.e.,
$$CPC^\dagger=QCPC^\dagger.$$ 
Right-multiplying by $  C  $ yields $  CP=QCP  $.
\end{proof}
In \quokka, the trace condition translates directly to a WMC equivalence.
For example, in Pauli-basis encoding, it is given as:
\begin{equation}
\label{eq:verif}
|F_{{\ket{\psi}\models P}}(\vec q)\wedge F_C(\vec q, \vec q') \wedge
F_{{\ket{\psi}\models Q}}(\vec q')|_W = |F_{{\ket{\psi}\models P}}(\vec q)|_W
\end{equation}
Here $F_{{\ket{\psi}\models P}}$ encodes those states that satisfy $P$, i.e.  $F_{\ket{\psi}\models \ketbra{0}{0}}(\vec q) = \no q_0$.
Note that the projector onto the $+1$ eigenspace of a Pauli string $\Pi_P$ is given by $P = \nicefrac12 (\Pi_P + I^{\otimes n})$.
Such projectors can be naturally encoded in the Pauli basis.
A variation of \autoref{eq:verif} for the computational-basis encoding is given as:
\[
\begin{aligned}
&|F_{C^\dagger}(\vec q_0, \vec q_1) \wedge F_P(\vec q_2,\vec q_3)
\wedge F_C(\vec q_4, \vec q_5)
\wedge F_Q(\vec q_5, \vec q_0)|_W 
&= |F_P(\vec q_0,\vec q_0)|_W.
\end{aligned}
\]

 \subsection{Equivalence checking}\label{sec:eqcheck}
A cornerstone of many quantum circuit optimization tasks~\cite{Bravyi2021cliffordcircuit,staudacher2023optimization} is \emph{equivalence checking}~\cite{hong2022equivalence,hong2021approximate,berent2022towards,advanced2021burgholzer,thanos2023fast}, as defined in~\autoref{func:equiv}.
Here, we use the \jfid{} between two unitaries \(U\) and \(V\), which is defined as
\(\mathit{Fid}_J(U,V)=2^{-2n}\lvert\Tr(UV^\dagger)\rvert^2\)~\cite{hong2021approximate}.

\begin{funct}[(Approximate) equivalence checking~\cite{hong2021approximate}]~\label{func:equiv}
    Given two $n$-qubit quantum circuits $C$, $C'$ and a real number $\epsilon\in[0,1]$,
    determine whether the \jfid{} between $C$ and $C'$ is greater than or equal to $1-\epsilon$.\\
When $\epsilon = 0$, we have exact equivalence checking.
\end{funct}

Fidelity checking between two \(n\)-qubit circuits can be reduced to computing the fidelity between \(UV^\dagger\) and the identity~\cite{tanaka2010exact,hong2021approximate}.
For example, the fidelity between \(U=T\cdot T\) and \(V=S\) is proportional to the weighted model count of the following Boolean formula:
\[
\underbrace{F_T(q_0,q_1) \wedge F_T(q_1,q_2)}_{\text{encoding of }U}\wedge\underbrace {F_{S^\dagger}(q_2,q_3)}_{\text{encoding of }V^\dagger}  \wedge 
 \underbrace{F_I(q_0,q_3)}_{\text{overlap with }I} \text{ where }
 F_I(q,q') = q \Leftrightarrow q'.
\]
We call this construction
the \emph{cyclic encoding}, 
as it essentially encodes $F_{UV^\dagger}(\vec q, \vec q)$.
It is formally justified by \cite[Th.~5]{zak2025syntesis} and supports both computational and Pauli basis: two distinct encodings differing in state and gate representation~\cite{mei2024disentangling}.
Since it encodes the fidelity between $U$ and $V$, it supports both exact and approximate equivalence checking. 
In contrast, the linear encoding~\cite{mei2024eq,thanos2023fast}, based on the Choi isomorphism~\cite{choi1975isomorphism}, only supports exact equivalence checking in the Pauli basis.
Overall, exact equivalence checking admits three encodings (linear in the Pauli basis and cyclic on both bases)~\cite{mei2024disentangling}, while fidelity computation is supported only by the cyclic encoding (in both bases).

 \subsection{Synthesis}\label{sec:syn}
Quantum processors can be based on different physical phenomena (trapped ions, photons, neutral atoms, superconducting qubits, etc), each with different properties in terms of qubit connectivity, gate sets, and operation costs~\cite{nielsen2000quantum}.
Synthesis~\cite{brand2023quantum,schneider2023sat,amy2023symbolic}  maps quantum algorithms to hardware with these various characteristics.
We define a circuit's \emph{depth} as the minimum number of time steps (or layers) needed to run the circuit, i.e., executing gates in parallel as much as possible.
Based on this, we define synthesis as follows.
\begin{funct}~\label{prob:c2c}
Given a quantum circuit $C_1$ in a finite gate set $\mathcal{G}_1$, a gate set $\mathcal{G}_2$, and an accuracy parameter $\epsilon \geq 0$,
find a {depth}-optimal quantum circuit $C_2$ in $\mathcal{G}_2$,
such that $C_1$ is approximately equivalent to $C_2$ up to $\epsilon$.
\end{funct}
Setting $\epsilon=0$ yields exact synthesis, 
but this is not always meaningful when $\mathcal{G}_1$ and $ \mathcal{G}_2$ yield different discretizations of the state space~\cite{giles2013exact}.
However, different \concept{universal} gate sets can always approximate each other up to arbitrary precision, as shown by the Solovay-Kitaev theorem~\cite {dawson2005solovaykitaev}.
In \quokka, the user can specify the gate set $\mathcal{G}_2$ based on the wide array of supported gates.
\quokka solves the synthesis problem via reduction to  weighted Max\#SAT~\cite{audemard2022maxsat},
which extends weighted model counting by allowing the definition of a set of optimization variables. It then returns an assignment to those variables that maximizes the formula's weight.  We first encode a one-layer parametric circuit $U_1(\vec q,\vec q',\vec y)$,
where assignments to $\vec y$ encode different gate choices in the layer.
For example, to synthesize a one-qubit circuit approximating $V=R_Z(\pi/8)$
using the gate set $\{H,T\}$, we encode
$
U_1(q,q',y) = (\neg y \rightarrow F_H(q,q')) \wedge (y \rightarrow F_T(q,q')) .
$

\quokka combines $U_1$ with $V$ into an approximate equivalence checking instance. It then invokes a weighted Max\#SAT call, with $\vec y$ as the optimization variables, to maximize the fidelity between $U_1$ and $V$.
If the maximum fidelity meets the threshold, the synthesized circuit is extracted from the
satisfying assignment to $\vec y$. Otherwise, it creates a new circuit $U_2$, that adds another layer to the construction of $U_1$.
This process iterates until the desired fidelity is achieved.
Correctness of this procedure is formally proved in~\cite[Proposition~11]{zak2025syntesis}.

Beyond predefined gate sets, \quokka also supports user-defined gate sets at the cost of higher runtime due to its general formulation. A user specifies a target gate set $\mathcal{G}_2$ and a fidelity threshold $1 - \epsilon$; \quokka then encodes each candidate gate via selector variables, so that the Max\#SAT solver jointly optimizes over gate choice and layer count. Supported built-in gates include $\{H, CX, S, T, CZ, C\sqrt{X}\}$, covering practically relevant gates.

  \section{Experimental Evaluation}
\label{sec:experiments}
We empirically evaluated \quokka on a machine equipped with an AMD Ryzen~9 7900X 12-Core processor, 24 hardware threads, and 62\,GB RAM. 
A time limit of 1800 seconds
is used for each benchmark instance.
We employ \gpmc~\cite{hashimoto2020gpmc,GPMCc} and
\ganak~\cite{srsm19} (and \dfourmax~\cite{lagniez2024d4}) 
as (Max)\#SAT solvers.
In addition to the main evaluation, we report case studies on verification
and synthesis.
All our encodings were accepted as benchmarks in the model counting competition~\cite{Competition2021_23}.

\myparagraphMajor{Simulation.}
We use the Munich Quantum Toolkit (MQT) benchmark set~\cite{mqt2023quetschlich}, focusing on
core algorithms and omitting minor variations (e.g., retaining QAOA but
excluding portfolio optimization with QAOA).
We additionally include benchmarks from RevLib, Feynman, and RevLib-H \cite{chen2025accelerating}; a translated version of RevLib that enables comparison with the gate set supported by \sliqsim.
All tools are tasked with computing the probability of measuring the
all-zero state.

We compare \quokka against decision diagram-based simulators
\ddsim~\cite{zulehner2019advanced},
\sliqsim~\cite{tsai2021bit}, and \quasimodo~\cite{quasimodo}.
We use the probability reported by DDSim as an independent reference value.
This comparison serves as a consistency check between independent implementations.
All reported answers are within $10^{-6}$ of this reference, except for the
Quasimodo result on the QNN instance with $n=18$ (marked as \ERROR), whose
reported probability differs from the reference by more than $10^{-6}$.
As shown in \autoref{tab:simulation_all}, the new computational-basis encoding in \quokka consistently outperforms the Pauli-basis encoding, often by orders of magnitude for both model counters.
Moreover, \quokka is competitive with \quasimodo, \sliqsim and \ddsim, particularly on circuits with many rotation gates, such as VQE, QFT, and W-state circuits.
By contrast, \ddsim performs better on highly structured circuits, such as the Grover, Qwalk, and several Feynman benchmarks.
A detailed discussion of the differences between above-mentioned tools and simulation tasks is given in \autoref{sec:discussion}.

\newcommand{\gpmcmark}{\textsuperscript{\textcolor{red}{+}}}
\newcommand{\ganakmark}{\textsuperscript{\textcolor{blue}{*}}}

\begin{table*}[!t]
\renewcommand{\arraystretch}{1.25}
\setlength{\tabcolsep}{5.5pt}
\caption{Simulation runtimes (sec) on MQTBench, Feynman, RevLib, and RevLib-H~\cite{chen2025accelerating} benchmarks, where $n$ ($|G|$) denotes the number of qubits (gates). MQTBench instances are grouped into rotation-heavy and structured families. {\timeout}, {\ERROR}, N/S, and {\CRASH} denote \emph{timeout}, \emph{wrong answer}, an unsupported benchmark, and a crash of \gpmc without producing an answer, respectively.}
\label{tab:simulation_all}
\centering
\scalebox{0.6}{
\begin{tabular}{p{3cm} | l |r r| r r  r r| r r r}
\hline\hline
\bf Suite
& \bf Circuit
& \multirow{2}{*}{$\mathbf{n}$}
& \multirow{2}{*}{$\mathbf{|G|}$}
& \multicolumn{2}{c}{\makecell[c]{\bf \quokka\\ (\toolFormat{Pauli})}}
& \multicolumn{2}{c|}{\makecell[c]{\bf \quokka \\(\toolFormat{Comp})}}
& \bf \multirow{2}{*}{\makecell[c]{\bf \toolFormat{Quasimodo}}}
& \bf \multirow{2}{*}{\makecell[c]{\bf \toolFormat{SliqSim}}}
& \bf \multirow{2}{*}{DDSim} \\
\cline{5-8}
& & & &
\textbf{\gpmc} & \textbf{\ganak} &
\textbf{\gpmc} & \textbf{\ganak} &
& & \\
\hline

\multirow{18}{=}{\makecell[l]{MQTBench\\(rotation-heavy)}}
& QAOA
& 14 & 266 & 0.23 & 2.27 & \fast{0.01} & 0.27 & 3.94 & N/S & 0.23 \\
& QAOA
& 15 & 285 & 0.84 & 3.69 & \fast{0.01} & 0.30 & 6.40 & N/S & 0.53 \\
& QAOA
& 16 & 304 & 0.60 & 2.61 & \fast{0.01} & 0.26 & 18.85 & N/S & 0.41 \\
\cline{2-11}

& QFT
& 16 & 672 & 0.02 & 0.06 & \fast{0.01} & 0.01 & 6.62 & N/S & 0.21 \\
& QFT
& 32 & 2,624 & 0.12 & 0.40 & \fast{0.03} & 0.05 & \timeout & N/S & 0.23 \\
& QFT
& 64 & 10,368 & 0.70 & 1.66 & \fast{0.13} & 0.14 & \timeout & N/S & 0.33 \\
\cline{2-11}

& QNN
& 18 & 1,259 & \timeout & \timeout & 182.53 & 79.03 & \ERROR & N/S & \fast{23.37} \\
& QNN
& 19 & 1,386 & \timeout & \timeout & \timeout & 275.40 & \timeout & N/S & \fast{133.96} \\
& QNN
& 20 & 1,519 & \timeout & \timeout & \timeout & \fast{598.93} & \timeout & N/S & 695.26 \\
\cline{2-11}

& VQE
& 14 & 236 & 0.11 & 1.62 & \fast{0.03} & 0.29 & 6.91 & N/S & 6.92 \\
& VQE
& 15 & 252 & 0.09 & 1.67 & \fast{0.03} & 0.33 & 20.32 & N/S & 22.66 \\
& VQE
& 16 & 270 & 0.09 & 1.90 & \fast{0.03} & 0.37 & 34.29 & N/S & 180.98 \\
\cline{2-11}

& W-state
& 32 & 559 & 0.11 & 3.24 & \fast{0.01} & 0.01 & \timeout & N/S & 0.22 \\
& W-state
& 64 & 1,135 & 0.32 & 6.65 & \fast{0.01} & 0.02 & \timeout & N/S & 0.23 \\
& W-state
& 128 & 2,287 & 1.39 & 16.58 & \fast{0.07} & 0.12 & \timeout & N/S & 0.44 \\
\cline{2-11}

& Su2random
& 16 & 744 & \timeout & \timeout & 89.07 & 62.81 & 76.68 & N/S & \fast{58.06} \\
& Su2random
& 17 & 816 & \timeout & \timeout & 138.87 & \fast{87.14} & 176.03 & N/S & 91.55 \\
& Su2random
& 18 & 891 & \timeout & \timeout & \timeout & \fast{248.00} & 462.18 & N/S & 461.20 \\
\hline

\multirow{15}{=}{\makecell[l]{MQTBench\\(structured)}}
& \makecell[l]{Qwalk (noancilla)}
& 8 & 9,369 & \timeout & \timeout & 4.51 & 34.92 & 2.64 & N/S & \fast{0.56} \\
& \makecell[l]{Qwalk (noancilla)}
& 9 & 18,585 & \timeout & \timeout & 17.91 & 745.47 & 9.39 & N/S & \fast{0.53} \\
& \makecell[l]{Qwalk (noancilla)}
& 10 & 37,017 & \timeout & \timeout & 68.89 & \timeout & 29.63 & N/S & \fast{0.76} \\
\cline{2-11}

& \makecell[l]{Grover (noancilla)}
& 14 & 235,379 & \CRASH & \timeout & \CRASH & \timeout & \timeout & 730.61 & \fast{31.19} \\
& \makecell[l]{Grover (noancilla)}
& 15 & 385,715 & \CRASH & \timeout & \CRASH & \timeout & \timeout & 1,259.86 & \fast{240.49} \\
& \makecell[l]{Grover (noancilla)}
& 16 & 630,354 & \CRASH & \timeout & \CRASH & \timeout & \timeout & \timeout & \timeout \\
\cline{2-11}

& DJ
& 512 & 2,001 & 0.06 & 0.21 & \fast{0.03} & 0.23 & 0.89 & 0.23 & 0.45 \\
& DJ
& 1024 & 4,067 & \fast{0.06} & 0.79 & \fast{0.06} & 0.86 & 1.07 & 2.41 & 1.10 \\
& DJ
& 2048 & 8,135 & 0.15 & 2.57 & \fast{0.17} & 2.73 & 1.48 & 10.77 & 5.04 \\
\cline{2-11}

& GHZ
& 512 & 512 & 0.02 & 0.18 & \fast{0.01} & \fast{0.01} & 0.76 & 0.06 & 0.39 \\
& GHZ
& 1024 & 1,024 & 0.04 & 0.79 & \fast{0.01} & 0.03 & 0.76 & 1.38 & 0.63 \\
& GHZ
& 2048 & 2,048 & 0.11 & 3.63 & \fast{0.02} & 0.07 & 0.79 & 3.25 & 1.97 \\
\cline{2-11}

& graphstate
& 512 & 1,024 & 0.02 & 0.05 & \fast{0.01} & 0.03 & 0.99 & 21.27 & 0.45 \\
& graphstate
& 1024 & 2,048 & 0.03 & 0.13 & \fast{0.02} & 0.09 & 1.30 & 560.27 & 0.85 \\
& graphstate
& 2048 & 4,096 & 0.07 & 0.12 & \fast{0.05} & 0.11 & 2.10 & 238.41 & 3.35 \\
\hline

\multirow{6}{*}{Feynman}
& ham15-low
& 17 & 213 & 14.97 & \timeout & 0.03 & 0.25 & 0.75 & \fast{0.01} & 0.22 \\

& ham15-med
& 17 & 452 & \timeout & \timeout & 0.08 & 2.20 & 0.77 & \fast{0.01} & 0.21 \\

& hwb10
& 16 & 31764 & \timeout & \timeout & \timeout & \timeout & 1.89 & 27.03 & \fast{0.31} \\

& hwb12
& 20 & 171482 & \CRASH & \timeout & \timeout & \timeout & 7.89 & 828.24 & \fast{0.88} \\

& gf2\textasciicircum 64\_mult
& 192 & 12731 & \timeout & \timeout & 20.36 & 63.73 & 1.40 & 5.36 & \fast{0.34} \\

& gf2\textasciicircum 256\_mult
& 768 & 198395 & \timeout & \timeout & \timeout & \timeout & 19.48 & \timeout & \fast{7.82} \\
\hline

\multirow{5}{=}{RevLib-H}
& add64\_184 & 193 & 385
& 0.09 & 9.86 & \fast{0.01} & \fast{0.01} & 0.78 & 0.02 & 0.27 \\

& e64-bdd\_295 & 195 & 516
& 0.10 & 1104.91 & \fast{0.01} & \fast{0.01} & 0.98 & 0.99 & 0.59 \\

& ex5p\_296 & 206 & 736
& \timeout & \timeout & \fast{0.01} & \fast{0.01} & 1.13 & 4.84 & 0.78 \\

& hwb9\_304 & 170 & 774
& \timeout & \timeout & \fast{0.01} & \fast{0.01} & 2.23 & 5.59 & 2.03 \\

& {nestedif2\_32\_445}
& 263 & 728 & N/S & N/S & N/S & N/S & N/S & \fast{146.05} & \timeout \\
\hline

\multirow{4}{*}{RevLib}
& urf2\_152
& 8 & 5030 & \timeout & \timeout & \fast{0.04} & 0.11 & 0.78 & 0.05 & 0.23 \\

& hwb9\_304
& 170 & 699 & \timeout & \timeout & \fast{0.01} & \fast{0.01} & 0.76 & \fast{0.01} & 0.25 \\

& add64\_184
& 193 & 256 & 0.35 & 27.54 & \fast{0.00} & \fast{0.00} & 0.76 & 0.01 & 0.25 \\

& ex5p\_296
& 206 & 647 & \timeout & \timeout & \fast{0.01} & \fast{0.01} & 0.77 & \fast{0.01} & 0.26 \\
\hline\hline

\end{tabular}
}
\end{table*}

\begin{table*}[!b]
\renewcommand{\arraystretch}{1.25}
\setlength{\tabcolsep}{2.4pt}
\caption{Equivalence checking runtimes (sec), where $n$ ($|G|$) denotes the number of qubits (gates).
{\timeout}, {\ERROR}, N/S, and {\CRASH} denote \emph{timeout}, \emph{wrong answer}, an unsupported benchmark, and a crash of \gpmc without producing an answer, respectively.}
\label{result_eq}
\centering

\begin{minipage}{0.495\textwidth}
\scalebox{0.6}{
\begin{tabular}{ c | r r | r r | r r | r | r }
  \hline\hline
  \bf Circuit & $\mathbf{n}$ & $\mathbf{|G|}$ 
  & \multicolumn{2}{c|}{\makecell[c]{\bf \quokka\\ \bf (\toolFormat{linear})}}
  & \multicolumn{2}{c|}{\makecell[c]{\bf \quokka\\ \bf (\toolFormat{cyclic})}}
  & \makecell[c]{\bf \qcec}
  & \makecell[c]{\bf \sliqec} \\
  \cline{4-7}
  & & & \textbf{\gpmc} & \textbf{\ganak} & \textbf{\gpmc} & \textbf{\ganak} & & \\
  \hline

\multirow{3}{*}{QAOA}
& 7 & 133
& 0.27 & 0.82 & \fast{0.25} & 0.70 & 0.08 & N/S \\
& 9 & 171
& {0.29} & \fast{0.74} & 223.83 & 14.72 & \ERROR & N/S \\
& 11 & 209
& {0.35} & 0.75 & \timeout & 525.04 & \fast{0.16} & N/S \\
\hline

\multirow{3}{*}{QFT}
& 2 & 14
& 0.25 & 0.25 & 0.25 & 0.35 & \fast{0.05} & N/S \\
& 8 & 176
& \fast{0.28} & \fast{0.28} & 277.23 & 336.85 & \ERROR & N/S \\
& 16 & 672
& \fast{0.42} & \fast{0.84} & \timeout & \timeout & \timeout & N/S \\
\hline

\multirow{3}{*}{QNN}
& 2 & 43
& \fast{0.26} & 0.55 & \fast{0.26} & 0.40 & \ERROR & N/S \\
& 8 & 319
& \CRASH & 533.58 & \timeout & \timeout & \fast{2.49} & N/S \\
& 16 & 1,023
& \timeout & \timeout & \timeout & \timeout & \timeout & N/S \\
\hline

\multirow{3}{*}{VQE}
& 5 & 83
& 0.26 & 0.61 & \fast{0.25} & 0.65 & \ERROR & N/S \\
& 10 & 168
& \fast{0.29} & 2.09 & 8.52 & 4.26 & \ERROR & N/S \\
& 15 & 253
& \fast{0.32} & 0.98 & 31.09 & 21.08 & 3.74 & N/S \\
\hline

\multirow{3}{*}{W-state}
& 16 & 271
& 0.33 & 0.61 & 0.31 & 1.01 & \fast{0.25} & N/S \\
& 32 & 559
& 0.54 & 2.62 & \fast{0.36} & 1.86 & \ERROR & N/S \\
& 64 & 1,135
& 1.47 & 8.81 & \fast{0.72} & 3.82 & \timeout & N/S \\
\hline\hline
\end{tabular}
}
\end{minipage}\begin{minipage}{0.495\textwidth}
\scalebox{0.6}{
\begin{tabular}{ c | r r | r r | r r | r | r }
  \hline\hline
  \bf Circuit & $\mathbf{n}$ & $\mathbf{|G|}$ 
  & \multicolumn{2}{c|}{\makecell[c]{\bf \quokka\\ \bf (\toolFormat{linear})}}
  & \multicolumn{2}{c|}{\makecell[c]{\bf \quokka\\ \bf (\toolFormat{cyclic})}}
  & \makecell[c]{\bf \qcec}
  & \makecell[c]{\bf \sliqec} \\
  \cline{4-7}
  & & & \textbf{\gpmc} & \textbf{\ganak} & \textbf{\gpmc} & \textbf{\ganak} & & \\
  \hline

\multirow{3}{*}{Su2random}
& 17 & 816
& \fast{4.17} & 6.00 & \timeout & \timeout & \timeout & N/S \\
& 18 & 891
& \timeout & \timeout & \timeout & \timeout & \timeout & N/S \\
& 19 & 969
& \fast{0.52} & 1.63 & \timeout & \timeout & \timeout & N/S \\
\hline

\multirow{3}{*}{\makecell[c]{Grover's \\ (noancilla)}}
& 5 & 499
& {5.97} & 62.17 & \fast{7.12} & 23.44 & \ERROR & N/S \\
& 6 & 1,568
& \timeout & \timeout & \timeout & \timeout & \fast{0.87} & N/S \\
& 7 & 3,751
& \timeout & \timeout & \timeout & \timeout & \fast{189.79} & N/S \\
\hline

\multirow{3}{*}{DJ}
& 128 & 1,007
& 3.70 & 4.77 & \fast{0.45} & 9.86 & \fast{0.45} & 9.92 \\
& 256 & 1,015
& 8.08 & 8.38 & \fast{0.59} & \timeout & 0.75 & 42.84 \\
& 512 & 2,001
& 37.70 & 31.93 & \fast{1.34} & \timeout & 2.24 & 240.26 \\
\hline

\multirow{3}{*}{GHZ}
& 128 & 130
& 1.61 & 1.79 & \fast{0.27} & 8.65 & 0.49 & 6.85 \\
& 256 & 256
& 4.38 & 4.34 & \fast{0.33} & 22.98 & 0.84 & 90.67 \\
& 512 & 512
& 26.40 & 25.33 & \fast{0.46} & \timeout & 1.72 & 372.79 \\
\hline

\multirow{3}{*}{graphstate}
& 128 & 1,280
& 8.70 & 6.91 & 4.14 & \fast{0.06} & \timeout & 13.22 \\
& 256 & 512
& 6.33 & 7.82 & \fast{0.38} & 13.96 & 0.94 & 392.07 \\
& 512 & 1,024
& 34.28 & \fast{0.51} & 38.94 & 32.44 & 2.36 & 1231.49 \\
\hline\hline
\end{tabular}
}
\end{minipage}

\end{table*}
 \myparagraphMajor{Equivalence checking.}
We perform exact equivalence checking on representative MQTBench families,
focusing on non-equivalent circuit pairs.
Non-equivalent instances are generated by first optimizing circuits with
PyZX~\cite{PyZXLS} and then injecting errors into the optimized circuits.
Specifically, we alter rotation angles to simulate \emph{phase-shift errors}, by injecting a small perturbation (adding $10^{-4}$) into a randomly chosen rotation gate. For cases without rotation gates,
we remove one gate randomly from the optimized circuit.
All our benchmarks consist of pairs of nonequivalent circuits, which favors \qcec because it performs many fast ``bug hunting'' approaches in parallel.

For \quokka, we use the cyclic encoding in the computational basis, as the linear encoding is not available there, and the linear encoding in the Pauli basis, which we found to outperform the cyclic encoding in that basis.
We compare against \sliqec~\cite{chen2022partial}, which uses decision diagrams, and \qcec v3.3.0~\cite{advanced2021burgholzer} from the MQT set, which combines decision diagrams and ZX calculus. 
Like \sliqsim, \sliqec offers exact results without floating-point error.

\autoref{result_eq} shows the results. \quokka outperforms \qcec on circuits with rotation gates (e.g., W-state, QFT). 
In exact equivalence checking, the proposed cyclic encoding serves as a complement to the linear encoding and outperforms it in certain cases (e.g., DJ and GHZ circuits). In contrast, approximate equivalence checking can only be performed using the cyclic encoding.
Notably, \quokka is significantly more robust than \qcec for phase-shift errors.
For all given benchmarks, the expected equivalence-checking result is known by construction (all circuit pairs are non-equivalent).
However, \qcec sometimes reports circuits as equivalent due to numerical instability. We mark such results as \ERROR.
Compared to \sliqec, \quokka shows better performance.

\myparagraphMajor{Verification.}
We consider the {quantum walk} circuit (5 qubits, 1305 gates) from
MQTBench and verify that an ancilla qubit is returned to its $\ket{0}$ state.
Both the pre- and post-condition are specified as $[\texttt{4:0}]$, referring
to the fifth qubit.

In both the computational and Pauli bases, the verification result is
\texttt{True}, with runtimes of 0.66 and 3.30\,seconds, respectively, confirming
that the ancilla qubit remains unchanged.
A direct comparison with \autoQ~\cite{chen2023autoq,chen2024autoq,chen2025verify}
is not applicable, as \autoQ models pre- and post-conditions as discrete sets
of basis states.
In contrast, our approach reasons over subspaces of quantum states, enabling
the specification of continuous families of superposition states; the two
predicate models are incomparable in expressive power.

\myparagraphMajor{Synthesis.}
\autoref{tab:syn} reports a case study on exact synthesis, decomposing the
Toffoli gate
$
\mathrm{CCX}=
\begin{smallmat}
I_{6\times 6} & 0_{6\times 2} \\
0_{2\times 6} & X
\end{smallmat}
$
into the Clifford+$C\sqrt{X}$ gate set $\{H,CX,C\sqrt{X}\}$.
Following the encoding in \autoref{sec:syn}, synthesis converges after six
iterations, yielding the circuit
$C\sqrt{X}_{0,2}C\sqrt{X}_{1,2}CX_{0,1}C\sqrt{X^\dagger}_{0,2}CX_{0,1}$,
where $CA_{i,j}$ denotes a controlled-$A$ gate with control qubit $i$ and target
qubit $j$.

Aside from the above gate set and Clifford+$T$~\cite{amy2013mitm}, our framework also supports user-defined gate sets at the cost of higher runtime due to its general formulation (see \autoref{sec:syn}). 

\begin{table}[h!] \renewcommand{\arraystretch}{1.25}
\setlength{\tabcolsep}{8pt}
\centering
\caption{Toffoli synthesis using \quokka with \texttt{d4Max}.}
\label{tab:syn}
\scriptsize
\resizebox{0.6\textwidth}{!}{
\begin{tabular}{lrrrrrr}
\hline\hline
\bf Iteration & 1 & 2 & 3 & 4 & 5 & 6 \\
\hline
\bf Variables & 24 & 60 & 96 & 132 & 168 & 204 \\
\bf Clauses & 136 & 992 & 1886 & 2780 & 3674 & 4568 \\
\bf Literals & 524 & 3760 & 7126 & 10492 & 13858 & 17224 \\
\bf Runtime & 0.03 & 0.04 & 0.16 & 3.39 & 141.86 & 577.78 \\
\hline
\end{tabular}
}
\vspace{-2em}
\end{table}
 
\section{Discussion}
\label{sec:discussion}

The aim of \quokka is to harness advances in satisfiability and model counting for quantum circuit analysis. This approach is informed by
how SAT solvers revolutionized e.g. model checking~\cite{biere2009bounded,bradley2011sat,een2011efficient,CTIGAR,donaldson2011software}.
They achieved this by avoiding the explicit construction of (strongest) fixpoints, as done in earlier approaches based on decision diagrams~\cite{bryant86,mcmillan,chaki2018bdd}.
In a similar vain, the model counting approach to quantum circuit analysis and optimization yields important distinctions from other tools, especially those based on decision diagrams (\ddsim, \quasimodo, \sliqsim, \qcec, and \sliqec), as discussed next.

\myparagraphMajor{Simulation Algorithm.}
In the first place, \quokka's algorithms can be understood as computing a path sum~\cite{mei2024disentangling} or ``Feynman simulation''~\cite{bernstein1993quantum,feynman2018simulating}. The other tools use ``Schr\"odinger simulation''~\cite{aaronson2017complexity} by explicitly computing the final quantum state, and all intermediate states, in some representation (decision diagrams).
These distinctions can be blurred, however.
Recent approaches have begun to combine the path sum with decision diagrams~\cite{wang2025feynmandd} and to integrate it more explicitly into the model counting method~\cite{huang2026equivalence}. There are also approaches that compute path sums using rewrite systems~\cite{amy2019towards} and knowledge compilation~\cite{huang2021logical}.

\myparagraphMajor{Representations.}
Symbolic formalisms may be able to succinctly represent the exponentially many amplitudes of the states encountered in the Schr\"odinger simulation.
Various forms of decision diagrams can be used for this purpose \cite{viamontes2003improving,sistla2023weighted,miller2006qmdd,vinkhuijzen2023limdd}, but also tensor networks or restricted Boltzmann machines, with different succinctness and tractability tradeoffs~\cite{vinkhuijzen2026knowledge}.

While model counters do not compute an explicit representation of individual quantum states, their trace can be understood as decision-DNNF~\cite{huang2007language,beame2013lower}, a normal form that subsumes decision diagrams. 
This trace implicitly represents the symbolic evolution described by the path sum, including the intermediate and final states.
If the trace is exported (a process that is typically called ``knowledge compilation''), the amplitudes and properties of these states can be queried on the decision-DNNF representation.
Decision-DNNF is exponentially more succinct~\cite{darwiche2002knowledge} than decision diagrams, but (naturally) retains the tractability of model counting on it, at the cost of sacrificing tractability for certain other operations (for instance, conjunction becomes hard for decision-DNNF).

\myparagraphMajor{Rounding.}
Computing with (negative) real numbers (the amplitudes of quantum states) inevitably introduces rounding errors in floating-point representations, unless one reverts to algebraic methods.
\sliqsim and \sliqec opt for an algebraic representation of the extended ring $\mathbb Z[i, \frac 1{\sqrt2}]$~\cite{giles2013exact}.
Such algebraic representations are known to have favorable asymptotic properties~\cite{quist2026exact} and both tools also show admirable performance, at the cost of being limited to Clifford+$T$ circuits (hence the omission of various benchmarks in \autoref{sec:experiments}).

The other tools studied here use floating-point representations.
For weighted model counting, rounding errors can pose significant challenges, but can also be reliably restrained for non-negative weights~\cite{bryant2025numerical}. Moreover, interval-based methods seem to provide sufficient guarantees in practice for negative weights~\cite{bryant2025numerical}.

In decision diagrams, weights have to be hashed to obtain a canonical form efficiently, causing another source of rounding errors~\cite{brand2025numerical}. This has proven a major obstacle in the past~\cite{zulehner2019efficiently,brand2025q,sanner2005affine}. For Quasimodo's CFLOBDD, this effect appears more limited in practice, likely due to the structure's reduced depth.

\myparagraphMajor{Task variation.}
\quokka targets exact and optimal quantum circuit analysis and optimization.
For instance, its synthesis method is depth-optimal, and the optional approximation provides provable guarantees~\cite{zak2025syntesis}. This approach is useful for specific tasks such as constructing new universal gate sets or improving magic gadgets to achieve better asymptotic results~\cite{bravyi2019simulation}, but it is limited in scale.
Heuristic approaches, for instance those based on ZX calculus~\cite{coecke2011interacting,villoria2026optimisation,van2020zx,kissinger2022simulating,PyZXLS}, enable the optimization of much larger circuits.

This paper considered strong simulation with computational- and
Pauli-basis measurements, which \quokka naturally supports in the Pauli-basis encoding.
The method, however, does not easily extend to weak simulation 
(see \autoref{table:encoding-comparison}), since existing CNF sampling algorithms~\cite{soos2020tinted,meel2022counting} break down in the presence of negative weights~\cite{mei2024disentangling}.
Recent extensions of \sliqsim support a broader range of measurement queries~\cite{sliqsimnew2025Chen}.
We believe that \quokka can support such queries by encoding the corresponding constraints directly in CNF.
Further, partial equivalence checking considers circuits with different qubit counts~\cite{chen2022partial}.
We believe that \quokka can also be extended to support this.

Different projectors can affect the performance of all methods.
This is best understood in the case of decision diagrams:
When the simulator succeeds in constructing the diagram $D$ representing the final state, computing the probability of obtaining outcome 0 (i.e., a measurement in the computational basis) on qubits $q_1,\dots,q_k$, i.e. a subsystem of the quantum states, reduces to (1) restriction $D[q_1:= 0,\dots,q_k:= 0]$ and (2) model counting (both take polynomial time in the size of the diagram~\cite{darwiche2002knowledge}).
However, if instead we ask for the probability of obtaining the outcome corresponding to $\ket{+}^{\otimes k}$ in a Hadamard-basis measurement of the same subsystem, the situation becomes markedly different.
One essentially has to perform second-level model counting (2AMC)~\cite{wang2024compositional} on the diagram.
Based on known results~\cite{wang2024compositional}, it seems that this is not tractable even for decision diagrams (OBDDs)
(unless the variables for qubits $q_1,\dots,q_k$ appear at the bottom).
This would show, as expected, that no method can avoid the inherent counting-like complexity in quantum computing~\cite{fortnow1999complexity}.
\quokka's approach, therefore, is to hand the entire problem to a \#SAT solver, which implicitly computes all steps at once (circuit evaluation and measurement/equivalence check/verification).

\section{Conclusion}
\label{sec:conclusion}

\quokka harnesses the strength of \#SAT solvers for quantum circuit simulation, equivalence checking, and depth-optimal synthesis into a single coherent tool.
In addition, this first tool release adds Hoare-logic verification, improved encodings (cyclic and computational basis), and approximate equivalence checking and synthesis.
Future extensions include automated encoding selection and incorporating fully symbolic algebraic model counting~\cite{kimmig2017algebraic,campos2020weightedmodelcounting}.

\quokka's impact on symbolic reasoning is two-fold.
First, \quokka indicates and stimulates improvements in (maximum) weighted
model counters.
So far, \quokka has motivated extensions to d4Max~\cite{audemard2022maxsat} and Ganak~\cite{srsm19}, adding support for negative, complex, and algebraic weights.
It also contributed new benchmarks to the model counting competition~\cite{Competition2021_23}.
Second, \quokka enables easy discovery of new applications of \#SAT to quantum computing, such as ancilla qubit uncomputation~\cite{paradis2021unqomp}.
Looking ahead, we expect \quokka to become a central tool for \#SAT-based analysis and optimization in quantum computing.

\myparagraphMajor{Acknowledgements.}
We thank the anonymous reviewers for their constructive feedback.
This work was supported by the Dutch National Growth Fund, as part of the Quantum Delta NL program.
The fourth author acknowledges the support received through the NWO Veni program (VI.Veni.232.381). 

\myparagraphMajor{Data Availability Statement.}
The implementation, benchmarks, scripts, and reproduction instructions are included
in the accompanying Zenodo artifact~\cite{quokka-artifact}. The Quokka\# source code
is also available on GitHub~\cite{quokka_sharp_github}.

{\fontsize{9}{11}\selectfont
\myparagraphMajor{Disclosure of Interests.}
The authors have no competing interests to declare.}

\bibliography{lit}

\end{document}